\documentclass[conference]{IEEEtran}
\usepackage{eurosym}
\usepackage{amssymb}
\usepackage{amsmath}
\usepackage{amsfonts}
\usepackage{float}
\usepackage{graphics}
\usepackage[tight,footnotesize]{subfigure}
\usepackage{rotating}
\usepackage{algorithm}
\usepackage{algorithmic}
\usepackage{enumerate}
\usepackage{cite}
\usepackage{xcolor}

\usepackage{color,soul}
\usepackage[colorinlistoftodos]{todonotes}

\usepackage{lipsum}

\newcommand*{\QEDA}
{\hfill\ensuremath{\blacksquare}}
\DeclareMathSizes{10}{9}{7}{5}

\newtheorem{theorem}{Theorem}

\input{tcilatex}
\begin{document}

\title{Fast Decoding of Low Density Lattice Codes}
\pubid{}
\specialpapernotice{}
\author{ Shuiyin Liu, Yi Hong, Emanuele Viterbo \\
ECSE Department, Monash University\\
Melbourne, VIC 3800, Australia\\
{shuiyin.liu, yi.hong, emanuele.viterbo@monash.edu} \and Alessia Marelli,
and Rino Micheloni \\
Flash Signal Processing Lab, Microsemi Corporation\\
Milan, 20127, Italy\\
{alessia.marelli, rino.micheloni@microsemi.com } \thanks{%
This work is supported by ARC under Grant Linkage Project No. LP160100002.} }
\maketitle

\begin{abstract}
Low density lattice codes (LDLC) are a family of lattice codes that can be
decoded efficiently using a message-passing algorithm. In the original LDLC
decoder, the message exchanged between variable nodes and check nodes are
continuous functions, which must be approximated in practice. A promising
method is Gaussian approximation (GA), where the messages are approximated
by Gaussian functions. However, current GA-based decoders share two
weaknesses: firstly, the convergence of these approximate decoders is
unproven; secondly, the best known decoder requires $O(2^d)$ operations at
each variable node, where $d$ is the degree of LDLC. It means that existing
decoders are very slow for long codes with large $d$. The contribution of this
paper is twofold: firstly, we prove that all GA-based LDLC decoders converge
sublinearly (or faster) in the high signal-to-noise ratio (SNR) region;
secondly, we propose a novel GA-based LDLC decoder which requires only $O(d)$
operations at each variable node. Simulation results confirm that
the error correcting performance of proposed decoder is the same
as the best known decoder, but with a much lower decoding complexity.


\end{abstract}


\section{Introduction}
Low density lattice codes (LDLC) is an efficiently decodable subclass of
lattice codes that approach the capacity of additive white Gaussian noise
(AWGN) channel \cite{sommer08}. Analogous to low density parity check codes
(LDPC), the inverse of LDLC generating matrix is sparse, allowing
message-passing decoding. LDLC is a natural fit for wireless communications
since both the codeword and channel are real-valued. Recently, LDLC has been
applied to many different communication systems, e.g., multiple-access relay
channel \cite{Chen15X}, half-duplex relay channels \cite{Chen16},
full-duplex relay channel \cite{Ferdinand15}, and secure communications \cite%
{Hooshmand16}. However, the basic disadvantage of LDLC is its high decoding
complexity, which limits its practical application. The commonly used LDLC
are relatively short, e.g., $1000$ in \cite{Chen15X,Chen16,Ferdinand15}. For
longer code length, the authors in \cite{Khodaiemehr17} suggest to use LDPC
lattice codes \cite{Sadeghi06} instead of LDLC, even though LDLC have better
performance than LDPC lattice codes.

The message-passing decoder of LDLC is slow because both the variable and
check nodes need to process continuous functions. The bottleneck occurs in
variable nodes, which have to compute the product of $d-1$ periodic continuous
functions, where $d$ is the degree of LDLC. This complicated operation
dramatically slows down the whole decoding process. To reduce the decoding
complexity, Gaussian approximation (GA)-based decoders have been proposed in
\cite{Kurkoski08,Yona09,Hernandez16}. The basic idea is to approximate the
messages exchanged between variable and check nodes by Gaussian functions.
The operation at each variable node reduces to compute the product of $d-1$
periodic Gaussian functions, ending up with a \emph{Gaussian mixture}.
However, it is costly to approximate a Gaussian mixture by a single Gaussian
function. Current approaches suggest to use the dominating Gaussian in the
mixture, by sorting \cite{Yona09} or exhaustive search \cite{Hernandez16}.
Even then, the computational complexity remains high, e.g., $O(2^d)$
in \cite{Hernandez16}. Since the value of $d$ is commonly set to $7$ \cite%
{sommer08}, current GA-based LDLC decoders are slow for long codes. The
other open question is, the convergence of GA-based LDLC decoders has not yet been proved.

The main contribution of this paper is twofold: first, we prove that in the
high signal-to-noise ratio (SNR) region, all GA-based LDLC decoders converge
\emph{sublinearly or faster}. This result verifies the goodness of Gaussian
approximation in LDLC decoding. Second, we propose a novel GA-based LDLC
decoder which requires only $O(d)$ operations at each variable node. The key
idea is to make use of the tail effect of Gaussian functions, i.e., if two
Gaussian functions have very distant means, then the product of them is approximately $0$. This fact allows us to approximate the Gaussian mixture by $2d-2$
Gaussian functions, without sorting  as in \cite{Yona09} or exhaustive search as in \cite{Hernandez16}. Simulation results
confirm that the performance of proposed decoder is the same as the best known one in \cite%
{Hernandez16}, give the same number of iterations. Note that having lower
decoding complexity enables us to run more iterations to further improve the
performance.


Section II presents the system model. Section III describes the convergence
of GA-based LDLC decoders. Section IV demonstrates the proposed decoder.
Section V shows the simulation results and comparisons with other decoders.
Section VI sets out the theoretical and practical conclusions. The Appendix
contains the proofs of the theorems.

\section{System Model}

\subsection{LDLC Encoding}

In an $N$-dimensional  lattice code, the codewords are defined by
\begin{equation}
\setlength{\abovedisplayskip}{3pt}
\setlength{\belowdisplayskip}{3pt}
\mathbf{x}=\mathbf{G}\mathbf{u},  \label{G}
\end{equation}
where $\mathbf{u} \in \mathbb{Z}^{N\times1}$ is an information integer
vector, $\mathbf{G} \in \mathbb{R}^{N\times N} $ is a real-valued {\em generator
matrix}, and $\mathbf{x} \in \mathbb{R}^{N\times1}$ is a real-valued
codeword. An LDLC is defined by a sparse {\em parity check matrix}, which is related to the generator matrix by
\begin{equation}
\setlength{\abovedisplayskip}{3pt}
\setlength{\belowdisplayskip}{3pt}
\mathbf{H}=\mathbf{G}^{-1},  \label{H}
\end{equation}
The sparsity of $\mathbf{H}$ enables the use of a
massage-passing algorithm to decode LDLC.

In the original construction of LDLC \cite{sommer08}, every row and column
in $\mathbf{H}$ has same $d$ non-zero values, except for random sign and
change of order. These $d$ values are referred to as \emph{generating sequence $\{h_1,\dots,h_d\}$ %
}, which is chosen in  \cite{sommer08}   as
\begin{equation}
\{h_1,\dots,h_d\} = \left\{ \pm1, \pm \frac{1}{\sqrt{d}}, ... , \pm \frac{1}{\sqrt{d}} \right\}.
\label{G_sequel}
\end{equation}
Note that there are only two distinct values, $1$ and $1/\sqrt{d}$. The
value of $d$ is referred to as the \emph{degree} of LDLC.

When a LDLC is used over an additive white Gaussian noise (AWGN) channel, we have
\begin{equation}
\setlength{\abovedisplayskip}{3pt}
\setlength{\belowdisplayskip}{3pt}
\mathbf{y}=\mathbf{x}+\mathbf{n},  \label{AWGN}
\end{equation}
where $\mathbf{n} \sim \mathcal{N}(0, \sigma^2 \mathbf{I})$ is a noise
vector, $\sigma^2$ is each dimension noise variance, and $\mathbf{I}$ is a $%
N-$dimensional unit matrix. The power of each LDLC codeword, denoted as $\|%
\mathbf{x}\|^2$, may be very large. Therefore, a \emph{shaping algorithm} is
required by LDLC, in order to make $\mathbf{x}$ distributed over a bounded
region, so called \emph{shaping region}. Various shaping methods have been
proposed in the literature \cite%
{Sommer09shaping,Ferdinand14shaping,Zhou17shaping}. In this work, we
consider the \emph{hypercube shaping} in \cite{Sommer09shaping}, where the
shaping region is a hypercube centered at the origin. Note that our results
are directly applicable to any shaping method.

\subsection{LDLC Decoding}

Since $\mathbf{H}$ is sparse, LDLC can be decoded by a message passing
algorithm \cite{sommer08}. The process takes four steps:

\begin{enumerate}
\item \emph{Initialization}: The $k^{\text{th}}$ variable node, denoted as $%
v_k$, sends a single Gaussian pdf $f_k(w)$ to its neighbor check nodes:
\begin{equation}
\setlength{\abovedisplayskip}{3pt}
\setlength{\belowdisplayskip}{3pt}
f_k(w)=\mathcal{N}(w; y_k, \sigma^2)=\frac{1}{\sqrt{2\pi\sigma^2}}e^{-\frac{%
(w-y_k)^2}{2\sigma^2}},  \label{initialize}
\end{equation}
for $k=1, ... , N$, where $y_k$ is the $k^{\text{th}}$ element in $\mathbf{y}
$.

\item \emph{Check-to-variable passing}: The $t^{\text{th}}$ check node,
denoted as $c_t$, sends a message $g_l(w)$ via its $l^{\text{th}}$ edge.
Without loss of generality, we assume that $c_t$ receives $f_i(w)$ from its $%
i^{\text{th}}$ edge, for $i=1, ... , d$. Let $h_i$ be the label of $i^{\text{%
th}}$ edge. The computation of $g_l(w)$ takes four steps:

\begin{enumerate}
\item \emph{convolution}: all $f_i(w/h_i)$, except $i = l$, are convolved:
\begin{eqnarray}
p_l(w) &=& f_1\left(\frac{w}{h_1}\right)\ast \cdots \ast f_{l-1}\left(\frac{w}{%
h_{l-1}}\right)  \notag \\
& & \ast f_{l+1}\left(\frac{w}{h_{l+1}}\right) \ast \cdots \ast f_{d}\left(%
\frac{w}{h_{d}}\right).  \label{convolution}
\end{eqnarray}

\item \emph{stretching}: the function $p_l(w)$ is stretched by $-h_l$:
\begin{equation}
\setlength{\abovedisplayskip}{3pt}
\setlength{\belowdisplayskip}{3pt}
\hat{p}_l(w)=p_l(-h_lw).  \label{stretching}
\end{equation}

\item \emph{periodic extension}: $\hat{p}_l(w)$ is extended to a periodic
function with period $1/|h_l|$:
\begin{equation}
\setlength{\abovedisplayskip}{3pt}
\setlength{\belowdisplayskip}{3pt}
g_l(w)=\sum^{\infty}_{i=-\infty}\hat{p}_l\left(w-\frac{i}{h_l}\right).
\label{periodic}
\end{equation}
\end{enumerate}

\item \emph{Variable-to-check passing}: The variable node $v_k$ sends a
message $f_j(w)$ via its $j^{\text{th}}$ edge. Similarly, we assume that $%
v_k $ receives $g_i(w)$ from its $i^{\text{th}}$ edge, for $i=1, ... , d$.
The computation of $f_j(w)$ takes two steps:

\begin{enumerate}
\item \emph{product}: all $g_i(w)$, except $i = j$, are multiplied:
\begin{equation}
\setlength{\abovedisplayskip}{3pt}
\setlength{\belowdisplayskip}{3pt}
\hat{f}_j(w)=\mathcal{N}(w; y_k, \sigma^2)\prod^{d}_{i=1,i \neq j}g_i(w).
\label{prod}
\end{equation}

\item \emph{normalization}: $\hat{f}_j(w)$ is normalized as:
\begin{equation}
\setlength{\abovedisplayskip}{3pt}
\setlength{\belowdisplayskip}{3pt}
f_j(w)=\frac{\hat{f}_j(w)}{\int^{\infty}_{-\infty}\hat{f}_j(w)dw}.
\label{normal}
\end{equation}
\end{enumerate}

Steps $2$ and $3$ are repeated until the desired number of iteration is
reached.

\item \emph{Final decision}: The variable node $v_k$ computes the product of
all received messages:
\begin{equation}
\setlength{\abovedisplayskip}{3pt}
\setlength{\belowdisplayskip}{3pt}
\hat{f}^{\text{final}}_k(w)=\mathcal{N}(w; y_k,
\sigma^2)\prod^{d}_{i=1}g_i(w).  \label{prod_final}
\end{equation}

The estimation of ${\mathbf{x}}=\{{x}_k\}$  in (\ref{AWGN}) and ${\mathbf{u}}$ in  (\ref{G})
are obtained by
\begin{equation}
\setlength{\abovedisplayskip}{3pt}
\setlength{\belowdisplayskip}{3pt}
\hat{x}_k=\arg \max_{w} \hat{f}^{\text{final}}_k(w).  \label{x_estimation}
\end{equation}
\begin{equation}
\setlength{\abovedisplayskip}{3pt}
\setlength{\belowdisplayskip}{3pt}
\hat{\mathbf{u}}=\lfloor\mathbf{H}\hat{\mathbf{x}} \rceil.
\label{u_hat_estimation}
\end{equation}
The operation $\left\lfloor \cdot \right\rceil$ rounds a number to the closest integer.
\end{enumerate}


According to (\ref{convolution})-(\ref{normal}), the messages exchanged
between variable and check nodes are continuous functions. In Sommer's implementation \cite{sommer08}, each continuous message is quantized and represented by a vector of $1024$ elements. The convolution phase at each check node, as well as the product phase at each variable node, have very high memory and computational requirements. This limits its application to relatively small dimensional LDLC.


\subsection{Gaussian-Approximation based LDLC Decoding}

Simplified decoding algorithms have been proposed in \cite%
{Kurkoski08,Yona09,Hernandez16}. The key idea is to approximate the variable
message $f_{j}(w)$ in (\ref{normal}) by a single Gaussian pdf:
\begin{equation}
\setlength{\abovedisplayskip}{3pt}
\setlength{\belowdisplayskip}{3pt}
f_{j}(w)\approx \mathcal{N}(w;m_{v,j},\sigma _{v,j}^{2}),  \label{MM}
\end{equation}%
where $m_{v,j}$ and $\sigma _{v,j}^{2}$ are the mean and variance of $%
f_{j}(w)$.

As a consequence, the check message $g_l(w)$ in (\ref{periodic}) reduces to a
periodic Gaussian pdf:
\begin{equation}
\setlength{\abovedisplayskip}{3pt}
\setlength{\belowdisplayskip}{3pt}
g_l(w)=\sum^{\infty}_{i=-\infty}\mathcal{N}_l\left(w; m_{c,l}-\frac{i}{h_l},
\sigma^2_{c,l}\right),  \label{Gaussian_periodic}
\end{equation}
where all component Gaussian pdfs in $g_l(w)$ have the same mean and
variance, denoted as $m_{c,l}$ and $\sigma^2_{c,l}$, respectively.


From (\ref{MM}) and (\ref{Gaussian_periodic}), we see that both variable and
check nodes only need to pass two values: the mean and variance of a
Gaussian function. This will greatly reduce the memory requirement for the messages. However, it is still
costly to perform the Gaussian approximation in (\ref{MM}). The problem lies
in the computation of the \emph{unnormalized} variable messages $\hat{f}_{j}(w)$,
which now reduce to
\begin{equation}
\setlength{\abovedisplayskip}{3pt}
\setlength{\belowdisplayskip}{3pt}
\hat{f}_{j}(w)=\mathcal{N}(w;y_{k},\sigma ^{2})\prod_{i=1,i\neq
j}^{d}\sum_{k=-\infty }^{\infty }\mathcal{N}_{i}\left( w;m_{c,i}-\frac{k}{%
h_{i}},\sigma _{c,i}^{2}\right) .  \label{Gaussian_prod}
\end{equation}%
which is a \emph{Gaussian mixture} of infinitely many components.

To simplify (\ref{Gaussian_prod}), the authors in \cite{Hernandez16} replace
each periodic Gaussian by only two Gaussians\footnote{In \cite{Hernandez16} the case using three Gaussians is also presented and it is shown to provides marginal performance improvements for a much higher complexity.} with a mean value close to $y_{k}$:
\begin{equation}
\setlength{\abovedisplayskip}{3pt}
\setlength{\belowdisplayskip}{3pt}
\hat{f}_{j}(w)\approx \mathcal{N}\prod_{i=1,i\neq j}^{d}\left( \mathcal{N}%
_{L,i}+\mathcal{N}_{R,i}\right) ,  \label{app_Gaussian_prod}
\end{equation}%
where $\mathcal{N}_{L,i}$ and $\mathcal{N}_{R,i}$ are the two Gaussian pdfs in the $i^{\text{th}}$ periodic Gaussian  with mean closest to $y_{k}$.
Recalling the fact that the product of Gaussian functions is still a single
Gaussian. The simplified $\hat{f}_{j}(w)$ can be written as a sum of $2^{d-1}$
Gaussian pdfs. This means that in each iteration, the computational
complexity at each variable node is proportional to $O(2^{d-1})$. Note that with
the value of $d=7$ used in \cite{sommer08}, there complexity is relatively large.

In summary, although the Gaussian-approximation (GA) based LDLC
decoders use less memory than the original decoder described in \cite{sommer08}, they are still too complex to be used for long codes. In Section IV we will propose a much faster decoder (still based on GA) to overcome these limitations.
Before that, we will tackle the other open
question whether it is possible to prove that GA-based LDLC decoders actually converge.

\section{Convergence Analysis of GA-Based LDLC Decoders}

In this section, we study the convergence speed of GA-based LDLC decoding
algorithms. Recalling that the original LDLC decoder has the following
property \cite{sommer08}:
\begin{eqnarray}
\lim_{K\rightarrow \infty }\sigma _{v,j}^{2} &=&0,\text{ if }|h_{j}|= 1/%
\sqrt{d}  \notag \\
\lim_{K\rightarrow \infty }\sigma _{v,j}^{2} &\leq &\theta ,\text{ if }%
|h_{j}|= 1  \label{convergence_op}
\end{eqnarray}%
where $K$ is the number of iterations and $\theta $ is a finite number. In
other words, each variable node generates $d-1$ \emph{narrow messages},
whose variance converges to $0$ as
$K\rightarrow \infty $. This implies the convergence of the pdf to Dirac centered at $m_{v,j}$:
\begin{equation}
\setlength{\abovedisplayskip}{3pt}
\setlength{\belowdisplayskip}{3pt}
\lim_{K\rightarrow \infty }f_{j}(w)=\delta(w-m_{v,j}),\text{ ~~~if }|h_{j}|= 1/\sqrt{d}
\label{convergence_mean}
\end{equation}
which ensures the convergence of the original LDLC decoder in \cite{sommer08}.

To study the impact of Gaussian approximation in (\ref{MM}) on convergence,
we also track the changes in variance $\sigma _{v,j}^{2}$ as $K$ increasing.
We have the following theorem.

\begin{theorem}
For all GA-based LDLC decoders with $d\geq 5$, at high SNR, the variances of
all variable messages satisfy
\begin{equation}
\sigma _{v,j}^{2}<\left\{
\begin{array}{cc}
\dfrac{1}{1.6K}\sigma ^{2}, & \text{ \ \ \ if }|h_{j}|= 1/\sqrt{d} \\[1.2ex]
\dfrac{2}{3}\sigma ^{2}, & \text{if }|h_{j}|= 1%
\end{array}%
\right.  \label{convegence_GA}
\end{equation}
\end{theorem}

\begin{proof}
See Appendix A.
\end{proof}

Theorem 1 shows that all GA-based LDLC decoders converge {sublinearly or
faster} at high SNR. This result demonstrates the goodness of Gaussian
approximation in LDLC decoding since each variable message also generates $d-1$
narrow messages. In practice, the bounds (\ref{convegence_GA}) is tight even
when SNR is close to the Shannon limit. An example is given bellow.

\begin{example}
We test a GA-based decoder in \cite{Hernandez16}. We consider a LDLC in \cite{LDLC961}, where $N=961$ and $d=7$. We tune the size of input alphabet such that the code rate $R$ is $2.8987$ bits/symbol. With hypercube shaping and uniform channel input, the Shannon capacity is $20.4$ dB. At SNR = $21.9$ dB, we simulate the average of the variances $\bar{\sigma}^2_{v,j}$ of all variable messages with $|h_{j}|=1/\sqrt{d}$. We define the ratio $\bar{\sigma}^2_{v,j}/\sigma^2$ as \emph{convergence speed}. In Fig. \ref{fig:1}, we compare the simulated convergence speed with the estimated one, i.e., $1/(1.6K)$ from (\ref{convegence_GA}), as a function of $K$. We see the bound is very tight even for finite SNRs.
\end{example}

\begin{figure}[tbp]
\centering \includegraphics[scale=0.65]{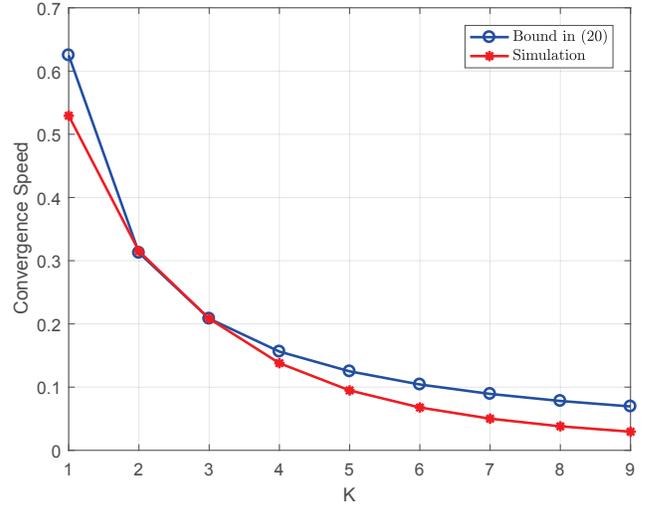} \vspace{-3
mm}
\caption{Convergence speed of GA-based LDLC decoder} \vspace{-3
mm}
\label{fig:1}
\end{figure}

\section{A Fast Decoding Algorithm of LDLC}

Theorem 1 shows that at high SNR, all GA-based decoders have a common  upper bound on the convergence speed. This implies that they will all have a similar performance at high SNR. The question is to find the most efficient GA-based decoder. In what follows we identify a GA-based decoder with a much lower complexity than the ones in the literature.



\subsection{Idea}
Recalling that the product of Gaussians returns a scaled single Gaussian,
e.g., for two Gaussians, we have
\begin{equation}
\setlength{\abovedisplayskip}{3pt}
\setlength{\belowdisplayskip}{3pt}
\mathcal{N}_{1}\left( w;m_{1},\sigma _{1}^{2}\right) \cdot \mathcal{N}%
_{2}\left( w;m_{2},\sigma _{2}^{2}\right) =c\mathcal{N}\left( w;m,\sigma
^{2}\right) ,  \label{prod_two_Gau}
\end{equation}%
where%
\begin{eqnarray}
\frac{1}{\sigma ^{2}} &=&\frac{1}{\sigma _{1}^{2}}+\frac{1}{\sigma _{2}^{2}}~,
~~~~~~~~
\frac{m}{\sigma ^{2}} ~~=~~\frac{m_{1}}{\sigma _{1}^{2}}+\frac{m_{2}}{\sigma
_{2}^{2}} \notag \\
c &=&\frac{1}{\sqrt{2\pi \left( \sigma _{1}^{2}+\sigma _{2}^{2}\right) }}%
\exp \left( -\frac{\left( m_{1}-m_{2}\right) ^{2}}{2\left( \sigma
_{1}^{2}+\sigma _{2}^{2}\right) }\right)   \label{c}
\end{eqnarray}
We have the following property:

\emph{Property 1:}We define $c$ as the \emph{height} of product of Gaussians, i.e., $ c\triangleq\mathcal{H}(\mathcal{N}_{1}\mathcal{N}_{2})$. Let the operation $\mathcal{P}(f(x))$ return the location of peak of a function $f(x)$, e.g., $\mathcal{P}(\mathcal{N}_{1})=m_1$. If the component Gaussians have very distant means/peaks (i.e., a large $|\mathcal{P}(\mathcal{N}_{1})-\mathcal{P}(\mathcal{N}_{2})|$), then $c\rightarrow 0$, and we can assume that their product is $\approx 0$. \QEDA 

This property allows us to simplify (\ref{app_Gaussian_prod}) by ignoring
a large number of vanishing products. Specifically, let $%
\mathcal{N}_{L,i}$ be the Gaussians with $\mathcal{P}(\mathcal{N}_{L,i}) \leq y_k$, and $\mathcal{N}_{R,i}$ with $\mathcal{P}(\mathcal{N}_{R,i}) \geq y_k$. We approximate (\ref{app_Gaussian_prod}) as follows:%
\begin{equation}
\setlength{\abovedisplayskip}{3pt}
\setlength{\belowdisplayskip}{3pt}
\hat{f}_{j}(w)\approx \mathcal{N}\prod_{i=1,i\neq j}^{d}\mathcal{N}_{L,i}+%
\mathcal{N}\prod_{i=1,i\neq j}^{d}\mathcal{N}_{R,i}.  \label{Proposed_App}
\end{equation}%
We ignore the cross-terms which
involve elements from both $\{\mathcal{N}_{L,i}\}^d_1$ and $\{\mathcal{N}_{R,i}\}^d_1$. For simplicity, we refer to the first term in (\ref{Proposed_App}) as \emph{left product}, and the second term as \emph{right product}.

To avoid the case where a cross-term has a greater height
than a left/right product, we need to select $\mathcal{N}_{L,i}$ and $\mathcal{N}_{R,i}$
carefully. Note that the crossover occurs when we have \emph{staggered pairs} of Gaussians. As shown in Fig. \ref{fig:2}, we consider
\begin{eqnarray*}
&&(\mathcal{N}_{L,1}+\mathcal{N}_{R,1})(\mathcal{N}_{L,2}+\mathcal{N}_{R,2})
\\
&=&\mathcal{N}_{L,1}\mathcal{N}_{L,2}+\mathcal{N}_{R,1}\mathcal{N}_{R,2}+%
\mathcal{N}_{L,1}\mathcal{N}_{R,2}+\mathcal{N}_{R,1}\mathcal{N}_{L,2}
\end{eqnarray*}%
where the two pairs of Gaussians are  staggered around the threshold, since $\mathcal{P}(\mathcal{N}_{L,1})$ is closer to $\mathcal{P}(\mathcal{N}_{R,2})$ than $\mathcal{P}(\mathcal{N}_{L,2})$. As a result, $\mathcal{H}(\mathcal{N}_{L,1}\mathcal{N}_{R,2})$ is greater than $\mathcal{H}(\mathcal{N}_{L,1}\mathcal{N}_{L,2})$ or $\mathcal{H}(\mathcal{N}_{R,1}\mathcal{N}_{R,2})$. It also means that both
$\mathcal{P}(\mathcal{N}_{L,1})$ and $\mathcal{P}(\mathcal{N}_{R,2})$ are close to $y_{k}$. Since the distance between $\mathcal{P}(\mathcal{N}_{L,i})$
and $\mathcal{P}(\mathcal{N}_{R,i})$ is either  $1$ or $\sqrt{d}$, then both $\mathcal{P}(\mathcal{N}_{L,2})$ and $\mathcal{P}(\mathcal{N}_{R,2})$ are far from $y_{k}$, at a distance up to $1$ or $\sqrt{d}$. This fact inspires us to break staggered pairs by deleting Gaussians whose means/peaks are far from $y_{k}$.

\emph{Left/Right Product Selection Criterion:} Consider the intervals $%
\mathcal{W}_i=\left[ y_{k}-\varepsilon_i,y_{k}+\varepsilon_i\right] $, where $0.5 <\varepsilon_i < \sqrt{d}$.

\begin{enumerate}
\item If $\mathcal{N}_{L,i}\in \mathcal{W}_i$ and $\mathcal{N}_{R,i}\notin
\mathcal{W}_i$, we select%
\begin{equation}
\setlength{\abovedisplayskip}{3pt}
\setlength{\belowdisplayskip}{3pt}
\mathcal{\hat{N}}_{L,i}=\mathcal{\hat{N}}_{R,i}=0.5\mathcal{N}_{L,i}.
\label{cri_1}
\end{equation}

\item If $\mathcal{N}_{L,i}\notin \mathcal{W}_i$ and $\mathcal{N}_{R,i}\in
\mathcal{W}_i$, we select%
\begin{equation}
\setlength{\abovedisplayskip}{3pt}
\setlength{\belowdisplayskip}{3pt}
\mathcal{\hat{N}}_{L,i}=\mathcal{\hat{N}}_{R,i}=0.5\mathcal{N}_{R,i}.
\label{cri_2}
\end{equation}

\item If $\mathcal{N}_{L,i}\in \mathcal{W}_i$ and $\mathcal{N}_{R,i}\in
\mathcal{W}_i$, we select%
\begin{equation}
\setlength{\abovedisplayskip}{3pt}
\setlength{\belowdisplayskip}{3pt}
\mathcal{\hat{N}}_{L,i}=\mathcal{N}_{L,i},\mathcal{\hat{N}}_{R,i}=\mathcal{N}%
_{R,i}.  \label{cri_3}
\end{equation}
\end{enumerate}

\noindent When $d=7$, we can set $\varepsilon_i=1$ for $|h_i|=1$, and $\varepsilon_i=1.7$ for $|h_i|=1/\sqrt{7}$.
A further discussion on the choice of $\varepsilon_i$ will be given in the journal version.
As a result, (\ref{Proposed_App}) is
updated to%
\begin{equation}
\setlength{\abovedisplayskip}{3pt}
\setlength{\belowdisplayskip}{3pt}
\hat{f}_{j}(w)\approx \mathcal{N}\prod_{i=1,i\neq j}^{d}\mathcal{\hat{N}}%
_{L,i}+\mathcal{N}\prod_{i=1,i\neq j}^{d}\mathcal{\hat{N}}_{R,i}.
\label{Proposed_App_Final}
\end{equation}%
A detailed explanation of proposed algorithm is given below.

\begin{figure}[tbp]
\centering \includegraphics[scale=0.8]{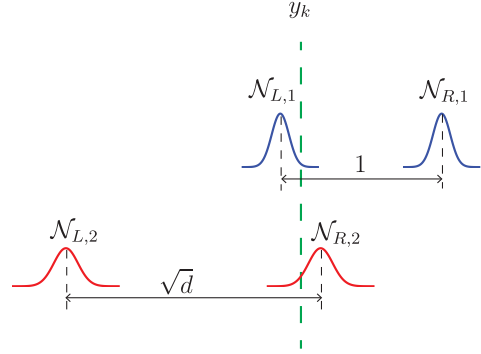} \vspace{-3
mm}
\caption{An example of staggered pairs of Gaussians} \vspace{-3
mm}
\label{fig:2}
\end{figure}

\subsection{Algorithm}

We only demonstrate the operation at each variable node, since the operation
at each check node is the same as \cite{Hernandez16}. The process takes
three steps:

\begin{enumerate}
\item \emph{Left/Right Product selection:} find $\mathcal{\hat{N}}_{L,i}$
and $\mathcal{\hat{N}}_{R,i}$, $i=1,...,d$, according to (\ref%
{cri_1})--(\ref{cri_3})

\item \emph{Mother message}: To avoid redundant computation, we compute the
left and right products from all inputs:%
\begin{equation}
\setlength{\abovedisplayskip}{3pt}
\setlength{\belowdisplayskip}{3pt}
\hat{f}_{\text{M}}(w)=\mathcal{N}\prod_{i=1}^{d}\mathcal{\hat{N}}_{L,i}+%
\mathcal{N}\prod_{i=1}^{d}\mathcal{\hat{N}}_{R,i},
\end{equation}%
which is referred to as the \emph{mother message}. Using (\ref{c}), we obtain
a sum of two scaled Gaussians:%
\begin{equation}
\setlength{\abovedisplayskip}{3pt}
\setlength{\belowdisplayskip}{3pt}
\hat{f}_{\text{M}}(w)=c_{L}\mathcal{N}_{L}+c_{R}\mathcal{N}_{R}\text{.}
\label{dominant_op}
\end{equation}

\item \emph{Individual message:} The message for the $i^{\text{th}}$ edge
can be obtained by subtracting $\mathcal{\hat{N}}_{L,i}$ and $\mathcal{\hat{N}}%
_{R,i}$ from $\hat{f}_{\text{M}}(w)$:%
\begin{equation}
\setlength{\abovedisplayskip}{3pt}
\setlength{\belowdisplayskip}{3pt}
\hat{f}_{j}(w)=c_{L,i}\mathcal{N}_{L,i}+c_{R,i}\mathcal{N}_{R,i}\text{.}
\end{equation}%
We normalize $\hat{f}_{j}(w)$ and apply Gaussian approximation%
\begin{eqnarray}
f_{j}(w) &=&\frac{c_{L,i}}{c_{L,i}+c_{R,i}}\mathcal{N}_{L,i}+\frac{c_{R,i}}{%
c_{L,i}+c_{R,i}}\mathcal{N}_{R,i}\text{ }  \notag \\
&\approx &\mathcal{N}(w;m_{v,j},\sigma _{v,j}^{2})\text{.}
\end{eqnarray}%
\end{enumerate}
\noindent {The values of $m_{v,j}$ and $\sigma _{v,j}^{2}$ are obtained from \cite[Eq. (2-3)]{Kurkoski10r}.} 


\subsection{Complexity}

The complexity of proposed variable node operation is dominated by (\ref%
{dominant_op}), which requires to apply (\ref{c}) $2d+2$ times. Therefore, the
complexity at each variable node is proportional to $O\left( d\right) $. This
is much lower than the best known decoder in \cite{Hernandez16}, which
requires $O(2^{d})$ operations at each variable node.

\section{Simulation Results}

This section compares the performance of the proposed LDLC decoder to the
best known decoder in \cite{Hernandez16}. Monte Carlo simulations
are used to estimate the symbol error rate (SER).

Fig. \ref{fig:3} shows the SER for LDLC with $d=7$ and $R=2.8987$ bits/symbol, using hypercube shaping. With hypercube shaping and uniform channel input, the Shannon capacity at $R=2.8987$ bits/symbol is $20.4$ dB. Both short and long codes are tested, i.e., $N=961$ and $N=10000$. Given $10$ iterations, the performance of proposed decoder coincides with the reference one
in both cases. Since the complexity of proposed decoder is much lower than the reference one,
we can run more iterations, e.g., $K=20$ for $N=10000$. In that case, the proposed decoder outperforms the reference one,
by about $0.8$ dB at SER $=10^{-5}$. Meanwhile, the gap to capacity (with cubic shaping) is about $0.2$ dB. This result confirms that our decoder works well for both short and long LDLC codes.

\begin{figure}[tbp]
\centering \includegraphics[scale=0.65]{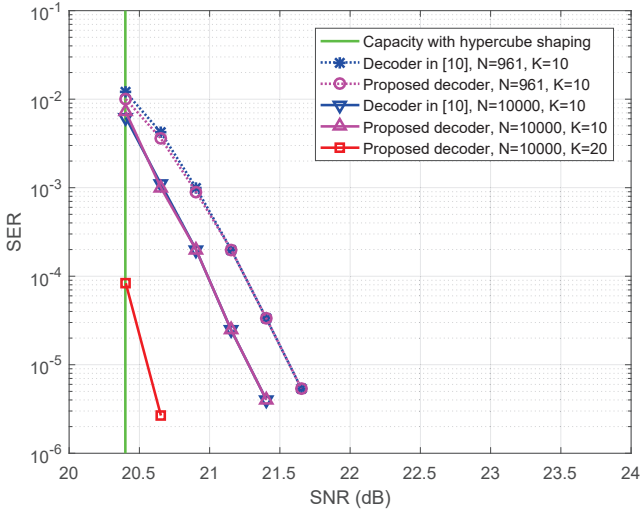} \vspace{-3
mm}
\caption{SER vs. SNR for the LDLC with $d=7$ and $R=2.8987$
bits/symbol} \vspace{-3 mm}
\label{fig:3}
\end{figure}

\section{Conclusions}

In this work, we have proved that all Gaussian-approximation based LDLC decoders have
the same convergence speed at high SNR. Inspired by this result, we proposed a fast
decoding algorithm which requires only $O(d)$ operations at each variable node. The new decoder
provides the same error correcting performance as the best known decoder,
but with much lower complexity. The proposed decoder enables the decode much longer LDLC which provide even better error correcting performance.


\section*{Appendix}

\subsection{Proof of Theorem 1}

Recalling the unnormalized variable message $\hat{f}_{j}(w)$ in (\ref%
{Gaussian_prod})
\begin{equation*}
\setlength{\abovedisplayskip}{3pt}
\setlength{\belowdisplayskip}{3pt}
\hat{f}_{j}(w)=\mathcal{N}(w;y_{k},\sigma ^{2})\prod_{i=1,i\neq
j}^{d}\sum_{k=-\infty }^{\infty }\mathcal{N}_{i}\left( w;m_{c,i}-\frac{k}{%
h_{i}},\sigma _{c,i}^{2}\right) .
\end{equation*}%
We assume that $\sigma ^{2}$ is small, i.e., $\mathcal{N}(w;y_{k},\sigma
^{2})$ is very narrow, such that for each periodic Gaussian, there will be
only one Gaussian component that contributes to the product. In this case,
the variance of normalized message equals to that of unnormalized one, since
they reduce to a single Gaussian.

In each iteration, we compute the variances of exchanged messages. At the $%
k^{\text{th}}$ iteration, let $C_{1}^{(k)}$ and $C_{\neq 1}^{(k)}$ be the
variances of check messages passed via the edges labeled by $\pm 1$ and $\pm
1/\sqrt{d}$, respectively. Similarly, let $V_{1}^{(k)}$ and $V_{\neq
1}^{(k)} $ be the variances of variable messages passed via the edges labeled
by $\pm 1 $ and $\pm 1/\sqrt{d}$, respectively.

Due to space limit, we ignore the computation for iterations 1 and 2. Full details will
be reported in the journal version.

%

\emph{Iteration 3}: For each check node, the variances are%
\begin{equation}
\setlength{\abovedisplayskip}{3pt}
\setlength{\belowdisplayskip}{3pt}
C_{\neq 1}^{(3)}=dV_{1}^{(2)}+(d-2)V_{\neq 1}^{(2)}<d\left(
V_{1}^{(1)}+V_{\neq 1}^{(1)}\right) <d\sigma ^{2},  \label{I3_1}
\end{equation}%
\begin{equation}
\setlength{\abovedisplayskip}{3pt}
\setlength{\belowdisplayskip}{3pt}
C_{1}^{(3)}=\frac{d-1}{d}V_{\neq 1}^{(2)}<V_{\neq 1}^{(2)}<\frac{1}{3.2}%
\sigma ^{2}.  \label{I3_2}
\end{equation}%
For each variable node, the variances are%
\begin{eqnarray}
V_{\neq 1}^{(3)} &=&\left( \frac{d-2}{C_{\neq 1}^{(3)}}+\frac{1}{C_{1}^{(3)}}%
+\frac{1}{\sigma ^{2}}\right) ^{-1}  \notag \\
&<&\left( \frac{1.6}{\sigma ^{2}}+\frac{1}{V_{\neq 1}^{(2)}}\right) ^{-1}<\frac{1}{4.8}\sigma ^{2},  \label{I3_3}
\end{eqnarray}%
\begin{equation}
\setlength{\abovedisplayskip}{3pt}
\setlength{\belowdisplayskip}{3pt}
V_{1}^{(3)}=\left( \frac{d-1}{C_{\neq 1}^{(3)}}+\frac{1}{\sigma ^{2}}\right)
^{-1}<\frac{2}{3}\sigma ^{2}.  \label{I3_4}
\end{equation}%
The bounds in (\ref{I3_1})-(\ref{I3_4}) can be extended to the rest
iterations:

\emph{Iteration }$k$: For each check node, the variances are%
\begin{equation}
\setlength{\abovedisplayskip}{3pt}
\setlength{\belowdisplayskip}{3pt}
C_{\neq 1}^{(k)}<d\sigma ^{2},
\end{equation}%
\begin{equation}
\setlength{\abovedisplayskip}{3pt}
\setlength{\belowdisplayskip}{3pt}
C_{1}^{(k)}<V_{\neq 1}^{(k-1)}.
\end{equation}%
For each variable node, the variances are%
\begin{equation}
\setlength{\abovedisplayskip}{3pt}
\setlength{\belowdisplayskip}{3pt}
V_{\neq 1}^{(k)}<\left( \frac{1.6}{\sigma ^{2}}+\frac{1}{V_{\neq 1}^{(k-1)}}%
\right) ^{-1}<\left( \frac{1.6}{\sigma ^{2}}+\frac{1.6\left( k-1\right) }{%
\sigma ^{2}}\right) ^{-1}=\frac{\sigma ^{2}}{1.6k}  \label{f1}
\end{equation}%
\begin{equation}
\setlength{\abovedisplayskip}{3pt}
\setlength{\belowdisplayskip}{3pt}
V_{1}^{(k)}<\frac{2}{3}\sigma ^{2}.  \label{f2}
\end{equation}%
Combining (\ref{f1}) and (\ref{f2}), we obtain (\ref{convegence_GA}). \QEDA




\bibliography{IEEEabrv,LIUBIB}
\bibliographystyle{IEEEtran}

\end{document}